\documentclass[10pt]{article}
\usepackage{graphicx,color}
\usepackage{amsmath}
\usepackage{amssymb}
\usepackage{amsmath,amsfonts,latexsym,graphicx,amssymb}
\usepackage[margin=0.7in, includefoot]{geometry}
\usepackage{mathrsfs}
\usepackage{amsthm}
\usepackage{dsfont}
\usepackage[utf8]{inputenc}
\usepackage{graphicx}
\usepackage[noadjust]{cite}
\newcommand{\ee}{\mathrm{e}}
\newcommand{\ii}{\mathrm{i}}
\renewcommand{\Im}{\mathrm{Im}}
\newcommand{\Tr}{\mathrm{Tr}}

\newcommand{{\Cd}}{{\mathbb{C}^d}}
\newcommand{{\Rn}}{{\mathbb{R}^n}}

\def\<{\langle}
\def\>{\rangle}
\newtheorem{Theorem}{Theorem}
\newtheorem{cor}{Corollary}

\newtheorem{definition}{Definition}

\newtheorem{pro}{Proposition}

\date{}

\begin{document}

\title{Unidirectional information flow and positive divisibility are nonequivalent notions of quantum Markovianity for noninvertible dynamics}

\author{\'Angel Rivas\\
{\small Departamento de Física Teórica, Facultad de Ciencias Físicas, Universidad Complutense, 28040 Madrid, Spain.}
\\ {\small CCS-Center for Computational Simulation, Campus de Montegancedo UPM, 28660 Boadilla del Monte, Madrid, Spain.}}

\maketitle

\begin{abstract}
We construct a dynamical map which is not positive divisible and does not present information backflow either (as measured by trace norm quantifiers). It is formulated for a qutrit system undergoing noninvertible dynamics. This provides an evidence that the two definitions of quantum Markovianity based on absence of information backflow and positive divisibility are nonequivalent for general noninvertible dynamical maps.
\end{abstract}

\bigskip

\begin{center}
{  \bf Dedicated to the memory of Prof. Andrzej Kossakowski}
\end{center}

\section{Introduction}

The study of memory effects in open quantum systems represents a very active area of research with important implications for quantum science and technology \cite{RHP_rev,Breuer_rev,Ines_rev, Michael_rev,MM}.  However, as a difference with the commutative classical case, the noncommutative algebra of observables associated to a quantum system makes nontrivial the formulation of a Markovian condition. Despite the first rigorous definitions of quantum Markovian systems were suggested in the 1980's \cite{Accardi,Lewis}, the notion was subject to a debate during years and there is still no a unique widely accepted criterion. Two of the most influential definitions are the ones based on absence of information backflow (as measured in terms of trace norm contractivity) \cite{BLP}, on the one hand, and on complete positive (CP) divisibility \cite{RHP}, on the other. More than years ago, Chru\'sci\'nski,  Kossakowski and Rivas studied the differences between these two notions of Markovianity \cite{CKR}.  In particular, it was stated that for invertible dynamical maps, both conditions are equivalent in a generalized sense: when the unidirectional information flow condition is formulated for generally biased discrimination problems and in the presence of extra systems acting as witness of the open system evolution. This equivalence was sharpen in \cite{Acin}, by assuming only unbiased discrimination problems at the expense of using extra systems with larger dimension than in \cite{CKR}. 

The problem of the equivalence between unidirectional information flow and CP-divisibility has been also studied for general noninvertible dynamics, albeit  the situation is considerably more intricate in such a case, and only partial results are known \cite{Acin,CRS,CC}. Noninvertible dynamics are typical of open quantum systems subjected to measurements, or when they are in contact with an environment made up of a finite number of degrees of freedom (see e.g. \cite{Erika,Jyrki,Si22} and references therein). Nevertheless, as fas as we know, there is no counterexample in the literature preventing from a general equivalence between unidirectional information flow and positive divisibility. The goal of the present work is to fill that gap. 

The counterexample we present here is inspired by a note by Matsumoto \cite{Matsumoto} focused on the failure of the Alberti-Uhlmann condition \cite{Alberti-Uhlmann} for quantum systems with dimension larger than 2 (see also \cite{Extending}). It is interesting to point the key role played by the continuity condition on dynamical maps in these equivalence problems, as relaxing continuity allows for a universal equivalence between CP-divisibility and unidirectional backflow of information in a  generalized sense \cite{BuscemiDatta}. Before explaining the counterexample we shall comment on the previous results in order to pose the problem in correct terms and fix our notation.

\section{Results on the equivalence between unidirectional information flow and positive divisibility}

We shall consider a quantum system with an associated $d$-dimensional Hilbert space $\mathcal{H}$,  and we denote by $\mathfrak{H}(\mathcal{H})$ the algebra of hermitian operators on $\mathcal{H}$.  In what follows, we shall restrict our attention to finite dimensional Hilbert spaces $d<\infty$.

\begin{definition}[Dynamical Map] We shall say that a one-parameter family of completely positive and trace preserving (CPTP) maps $\{\Lambda_t\}_{t\geq0}$ on $\mathfrak{H}(\mathcal{H})$ forms a \emph{dynamical map} if the correspondence $t\mapsto\Lambda_t$ is continuous and $\Lambda_0=\mathds{1}$.  
\end{definition}

According to Breuer, Laine and Piilo (BPL) \cite{BLP}, quantum Markovianity is identified with the lack of information backflow as quantified by
\begin{equation}\label{BLP}
\frac{d}{dt}\|\Lambda_t(\rho_1)-\Lambda_t(\rho_2)\|_1\leq 0,
\end{equation}
for all pairs $\rho_{1,2}$ of density matrices in $\mathfrak{H}(\mathcal{H})$. Here, trace distance is used, which is the metric induced by the trace norm $\|X\|_1:=\Tr|X|=\Tr\sqrt{X^\dagger X}$ for $X\in \mathfrak{H}(\mathcal{H})$.  Throughout the paper, we shall assume right derivatives in case of nonequivalent left and right limits in expressions like \eqref{BLP}:
\begin{equation}
\frac{d}{dt}f(t):=\lim_{\epsilon\downarrow0}\frac{f(t+\epsilon)-f(t)}{\epsilon}.
\end{equation}

On the other hand, Rivas, Huelga and Plenio (RHP) \cite{RHP} suggested that a dynamical map $\{\Lambda_t\}_{t\geq0}$ should be called Markovian if and only if it is CP-divisible, i.e.  it admits the decomposition
\begin{equation}\label{Divisibledecomposition}
\Lambda_t=\Lambda_{t,s}\Lambda_s, \quad t>s>0,
\end{equation}
with $\Lambda_{t,s}$ another CPTP map. If $\Lambda_{t,s}$ is not completely positive but just positive, the dynamics is called P-divisible. Thus CP-divisible maps are P-divisible but the opposite is not true.

The BLP approach is connected to the one-shot discrimination problem. Namely, the quantity $\frac{1}{2}-\tfrac{1}{4}\|\rho_1-\rho_2\|_1$ is the minimum average error probability of when trying to discriminate, by performing one single measurement, whether the quantum system was prepared in the state $\rho_1$ or in $\rho_2$ with equal prior probabilities $p_1=p_2=1/2$. Therefore, if \eqref{BLP} is violated for some $t$, the transitory increment in the trace distance between $\Lambda_t(\rho_1)$ and $\Lambda_t(\rho_2)$ can be identified as a backflow of information from environment to the system. Otherwise, the flow of information is always unidirectional from system to environment, leading to a degradation in the capability to discriminate between $\rho_1$ and $\rho_2$. Following the same reasoning, one can extend the approach to account for information backflow in generally biased discrimination problems, where the prior probabilities for preparing $\rho_1$ or $\rho_2$ are different $p_1\neq p_2$. In such a case, the minimum average error probability becomes $\frac{1}{2}-\tfrac{1}{2}\|p_1\rho_1-p_2\rho_2\|_1$  \cite{Helstrom,Hayashi}. Thus, a generalized BLP condition for quantum Markovianity reads
\begin{equation}\label{BLP2}
\frac{d}{dt}\|\Lambda_t(p_1\rho_1-p_2\rho_2)\|_1\leq 0,
\end{equation}
for any probabilities $p_{1,2}$, $p_1+p_2=1$ and density matrices $\rho_{1,2}\in\mathfrak{H}(\mathcal{H})$. This is equivalent to monotonic contractivity of $\Lambda_t$ in $\mathfrak{H}(\mathcal{H})$, 
\begin{equation}\label{BLP3}
\frac{d}{dt}\|\Lambda_t(X)\|_1\leq 0, \quad X\in\mathfrak{H}(\mathcal{H}).
\end{equation}
This condition can be connected to the divisibility approach. In particular, Chru\'sci\'nski, Kossakowski and Rivas, showed the following result.

\begin{Theorem}[\cite{CKR}] If $\{\Lambda_t\}_{t\geq0}$ is an invertible dynamical map, i.e. $\Lambda_{t}^{-1}$ exists for $t\geq0$, then the condition \eqref{BLP3} is equivalent to P-divisibility of $\{\Lambda_t\}_{t\geq0}$, and, in addition, $\{\Lambda_t\}_{t\geq0}$ is CP-divisible if and only if
\begin{equation}\label{CKR}
\frac{d}{dt}\|\Lambda_t\otimes\mathds{1}(X)\|_1\leq 0, \quad X\in\mathfrak{H}(\mathcal{H}\otimes\mathcal{H} ). 
\end{equation}
\end{Theorem}

There are other notable ways to establish the equivalence between CP-divisible and monotonic contractivity for invertible dynamical maps:

\begin{Theorem}[\cite{Acin}] If $\{\Lambda_t\}_{t\geq0}$ is an invertible dynamical map, i.e. $\Lambda_{t}^{-1}$ exists for $t\geq0$, then it is CP-divisible if and only if
\begin{equation}
\frac{d}{dt}\|\Lambda_t\otimes\mathds{1}_{d+1}(\rho_1-\rho_2)\|_1\leq 0, 
\end{equation}
for all density matrices $\rho_{1,2} \in\mathfrak{H}(\mathcal{H}\otimes\mathcal{H}' )$, with $d=\dim \mathcal{H}=\dim \mathcal{H}'-1$.
\end{Theorem}
This theorem reduces the norm contractivity condition \eqref{CKR} to trace distance contractivity between density matrices as in the initial BLP formulation \eqref{BLP}, at the expense of increase the dimension of the second subsystem by one. 

\begin{Theorem}[\cite{FabioDarek}] \label{TheoTensorProduct} If $\{\Lambda_t\}_{t\geq0}$ is an invertible dynamical map, i.e. $\Lambda_{t}^{-1}$ exists for $t\geq0$, then it is CP-divisible if and only if
\begin{equation}
\frac{d}{dt}\|\Lambda_t\otimes\Lambda_t(X)\|_1\leq 0, \quad X\in\mathfrak{H}(\mathcal{H}\otimes\mathcal{H} ). 
\end{equation}
\end{Theorem}
Note that the condition $\|\mathcal{E}\otimes\mathcal{E}(X)\|\leq \|X\|$ is not enough for $\mathcal{E}$ to be CPTP (the transposition is a clear counterexample), so Theorem \ref{TheoTensorProduct} shows the nontrivial role played by the continuity condition of a dynamical map.

The connection between monotonic norm contractivity and divisibility properties was also formulated for general $k$-divisible dynamical maps. Namely, those that admit a decomposition such as \eqref{Divisibledecomposition}  with $\Lambda_{t,s}$ $k$-positive. Thus $k$-divisible dynamical maps lie in between P-divisible maps or 1-divisible, and CP-divisible maps or $d$-divisible. 

\begin{Theorem}[\cite{DarekSabrina}] \label{TheoDarekSabrina}
 If $\{\Lambda_t\}_{t\geq0}$ is an invertible dynamical map, i.e. $\Lambda_{t}^{-1}$ exists for $t\geq0$, then it is $k$-divisible if and only if
\begin{equation}
\frac{d}{dt}\|\Lambda_t\otimes\mathds{1}_k(X)\|_1\leq 0, \quad X\in\mathfrak{H}(\mathcal{H}\otimes\mathcal{H}' ),
\end{equation}
with $\dim\mathcal{H}'=k$.
\end{Theorem}

The situation for noninvertible dynamical maps is considerable more difficult. However, Theorem \ref{TheoDarekSabrina} admits a generalization for the so-called image nonincreasing dynamical maps.

\begin{Theorem}[\cite{CRS}]\label{TheoremCRS}
If $\{\Lambda_t\}_{t\geq0}$ is an image nonincreasing dynamical map, i.e. $\Im(\Lambda_{t})\subseteq \Im(\Lambda_{s})$ for $t\geq s$, then it is $k$-divisible if and only if
\begin{equation}\label{CRS_k}
\frac{d}{dt}\|\Lambda_t\otimes\mathds{1}_k(X)\|_1\leq 0, \quad X\in\mathfrak{H}(\mathcal{H}\otimes\mathcal{H}' ),
\end{equation}
with $\dim\mathcal{H}'=k$.
\end{Theorem}

The proof of this theorem follows the same steps as the one for the CP-divisibility case in \cite{CRS}. Clearly, invertible dynamical maps are a special instance of image nonincreasing dynamical maps. Other important kinds of image nonincreasing dynamical maps are diganonalizable commutative dynamical maps (here commutative means $\Lambda_t\Lambda_s=\Lambda_s\Lambda_t$ for all $t$ and $s$) and normal dynamical maps, which satisfy $\Lambda_t^\dagger \Lambda_t= \Lambda_t\Lambda_t^\dagger$ with $\Lambda_t^\dagger$ is the Heisenberg picture adjoint $\Tr[X\Lambda(Y)]=\Tr[\Lambda^\dagger(X)Y]$, $X,Y\in\mathfrak{H}(\mathcal{H})$. 

The previous theorem can be trivially extended to dynamical maps unitarily equivalent to image nonincreasing ones, as the trace norm is invariant under unitary maps.

\begin{cor} \label{CorollaryCRS} Suppose that $\{\Lambda_t\}_{t\geq0}$ is unitarily equivalent to an image nonincreasing dynamical map, i.e. there exists some unitary  $\mathcal{U}_t(X)=U_t X U_t^\dagger$ for $X\in\mathfrak{H}(\mathcal{H})$ such that $\{\mathcal{U}_t\Lambda_t\}_{t\geq0}$ is image nonincreasing. Then $\{\Lambda_t\}_{t\geq0}$ is $k$-divisible if and only if \eqref{CRS_k} holds.
\end{cor}

Beyond image nonincreasing (or unitarily equivalent) dynamical maps the equivalence has been only established for CP-divisibility of qubits.
\begin{Theorem}[\cite{CC}] \label{TheoremCC} A dynamical map $\{\Lambda_t\}_{t\geq0}$ on $\mathfrak{H}(\mathcal{H})$ with $\dim\mathcal{H}=2$ is CP-divisible if and only if
\begin{equation}\label{CC}
\frac{d}{dt}\|\Lambda_t\otimes\mathds{1}(X)\|_1\leq 0, \quad X\in\mathfrak{H}(\mathcal{H}\otimes\mathcal{H} ). 
\end{equation}
\end{Theorem}

At this point, it is not known whether the equivalence between some kind of positive divisibility and unidirectional information flow conditions can be extended further than image nonincreasing dynamical maps and/or quantum systems of dimension larger than 2. In the following section we present a partial negative result, by constructing a dynamical map which satisfies condition \eqref{BLP3} but it is neither CP-divisible nor P-divisible.

\section{Counterexample}

Let $\{|1\rangle,|2\rangle,|3\rangle\}$ be an orthonormal basis of the Hilbert space of a qutrit, $d=3$. Consider the following CP maps written in terms of this basis:
\begin{align}
\mathcal{E}_1(X)&=\tfrac{1}{4}(X+D_1 X D_1+D_2 X D_2+ D_3 X D_3),\label{E1}\\
\mathcal{E}_2(X)&=K_2 X K_2^\dagger,\\
\mathcal{E}_3(X)&= \langle 1|X|1\rangle \rho_{A}+\langle 2|X|2\rangle\rho_{B},\\
\mathcal{E}_4(X)&= 2\langle 1|X|1\rangle |1\rangle \langle 1| + 2\langle 2|X|2\rangle |\theta \rangle \langle \theta |, \quad |\theta\rangle:=\ee^{\ii G \theta}|2\rangle
\end{align}
with 
\begin{equation}
D_1=\left(\begin{smallmatrix}
-1 &  &  \\
 & 1 &  \\
 &  & 1
\end{smallmatrix}\right), \quad D_2=\left(\begin{smallmatrix}
1 &  &  \\
 & -1 &  \\
 &  & 1
\end{smallmatrix}\right), \quad D_3=\left(\begin{smallmatrix}
1 &  &  \\
 & 1 &  \\
 &  & -1
\end{smallmatrix}\right), \quad K_2=\left(\begin{smallmatrix}
1 & 0 & 0 \\
0 & 1 & 1 \\
0 & 0 & 0
\end{smallmatrix}\right),
\end{equation} 
\begin{equation}
\rho_A=\frac12\left(\begin{smallmatrix}
1 &  &  \\
 & 0 &  \\
 &  & 1
\end{smallmatrix}\right), \quad \rho_B=\frac12\left(\begin{smallmatrix}
0 &  &  \\
 & 1 &  \\
 &  & 1
\end{smallmatrix}\right), \quad  \text{and}\quad G=\left(\begin{smallmatrix}
0 & -\ii & 0\\
\ii & 0 & 0\\
0 & 0 & 0
\end{smallmatrix}\right).
\end{equation}
Furthermore, consider also the following $\tau$-parametrized families of CP maps, with $\tau\in[0,1]$,
\begin{align}
\Gamma_\tau^{(1)}(X)&=\ee^{\int_{0}^\tau \mathcal{L}_{s} ds},\label{Gamma1}\\
\Gamma_\tau^{(2)}(X)&=\ee^{-f_1(\tau)}X+[1-\ee^{-f_1(\tau)}]\mathcal{E}_2(X),\label{Gamma2}\\
\Gamma_\tau^{(3)}(X)&=\ee^{-f_2(\tau)}X+[1-\ee^{-f_2(\tau)}]\mathcal{E}_3(X),\label{Gamma3}\\
\Gamma_\tau^{(4)}(X)&=(1+\tau^2)\Big[\langle 1|X|1\rangle |1\rangle \langle 1|+\langle 2|X|2\rangle \ee^{\ii G \theta \tau}|2\rangle \langle 2|\ee^{-\ii G \theta \tau}\Big] + (1-\tau^2)\big(\langle 1|X|1\rangle +\langle 2|X|2\rangle \big) |3\rangle \langle 3|,\label{Gamma4}
\end{align}
and the $s$-dependent GKLS generator \cite{GKLS}
\begin{equation}\label{Ls}
\mathcal{L}_s(X):=\gamma(s)[D_1 X D_1+D_2 X D_2+D_3 X D_3-3X ].\\
\end{equation}
In these definitions, $f_{1,2}(\tau)$ are derivable and monotonically increasing functions which satisfy
\begin{equation}\label{f12}
0=f_{1,2}(0)< f_{1,2}(\tau)< f_{1,2}(1)=\infty,
\end{equation}
and $\gamma(s)$  is a continuous, positive bounded function for $s\in(0,1)$, but singular at $s=1$.  It is easy to check that 
\begin{equation}\label{limGamma1}
\lim_{\tau\to1}\Gamma_{\tau}^{(1)}=\mathcal{E}_1 .
\end{equation}
Then,  we construct the following family of maps:
\begin{equation}\label{c-exem}
\Lambda_t=\begin{cases} \Gamma_{t/t_1}^{(1)}, & 0\leq t< t_1,\\
\Gamma_{(t-t_1)/(t_2-t_1)}^{(2)} \mathcal{E}_1  & t_1\leq t< t_2,\\
\Gamma_{(t-t_2)/(t_3-t_2)}^{(3)} \mathcal{E}_2\mathcal{E}_1  & t_2\leq t<t_3,\\
\Gamma_{(t-t_3)/(t_4-t_3)}^{(4)} \mathcal{E}_3\mathcal{E}_2\mathcal{E}_1 & t_3\leq t\leq t_4.
\end{cases}
\end{equation}
The specific values of $t_1$, $t_2$, $t_3$ and $t_4$ and the form of $\Lambda_t$ for $t>t_4$ is irrelevant in what follows. 

\begin{pro} The family $\{\Lambda_t\}_{t\in[0,t_4]}$ in \eqref{c-exem} is continuous.
\end{pro}

\begin{proof} We will check continuity in $t=t_1,t_2,t_3$, as it clearly holds for the rest of the points.  From  \eqref{Gamma2},  \eqref{f12} and \eqref{limGamma1}, we have
\[
\lim_{t\uparrow t_1}\Lambda_t=\lim_{\tau\uparrow 1}\Gamma_{\tau}^{(1)}=\mathcal{E}_1=\lim_{\tau\downarrow 0}\Gamma_{\tau}^{(2)}\mathcal{E}_1=\lim_{t\downarrow t_1}\Lambda_t ,
\]
and, as a straightforward consequence of \eqref{Gamma2}, \eqref{Gamma3} and \eqref{f12},
\[
\lim_{t\uparrow t_2}\Lambda_t=\lim_{\tau\uparrow 1}\Gamma_{\tau}^{(2)}\mathcal{E}_1=\mathcal{E}_2\mathcal{E}_1=\lim_{\tau\downarrow 0}\Gamma_{\tau}^{(3)}\mathcal{E}_2\mathcal{E}_1=\lim_{t\downarrow t_2}\Lambda_t .
\]
Furthermore, 
\[
\lim_{\tau\downarrow 0}\Gamma_{\tau}^{(4)}=\Gamma_{0}^{(4)}=\langle 1|X|1\rangle (|1\rangle \langle 1|+|3\rangle \langle 3|)+\langle 2|X|2\rangle (|2\rangle \langle 2|+|3\rangle \langle 3|),
\]
and we obtain 
\[
\Gamma_{0}^{(4)}(\rho_A)=\rho_A \quad \text{and} \quad \Gamma_{0}^{(4)}(\rho_B)=\rho_B.
\]
Therefore, $\Gamma_{0}^{(4)}\mathcal{E}_3=\mathcal{E}_3$, and continuity in $t_3$ also holds
\[
\lim_{t\uparrow t_3}\Lambda_t=\lim_{\tau\uparrow 1}\Gamma_{\tau}^{(3)}\mathcal{E}_2\mathcal{E}_1=\mathcal{E}_3\mathcal{E}_2\mathcal{E}_1=\lim_{\tau\downarrow 0}\Gamma_{\tau}^{(4)}\mathcal{E}_3\mathcal{E}_2\mathcal{E}_1=\lim_{t\downarrow t_3}\Lambda_t.
\]
\end{proof}

\begin{pro} The family $\{\Lambda_t\}_{t\in[0,t_4]}$ in \eqref{c-exem} is trace preserving.
\end{pro}

\begin{proof} $\Lambda_t$ is clearly trace preserving for $0\leq t\leq t_1$, as $\mathcal{L}_s$ in \eqref{Ls} is a (time-dependent) GKSL form. In addition, we just note that $\Gamma_{\tau}^{(i+1)}$ is trace preserving on $\Im(\Lambda_{t_i})$ for $i=1,2,3$, respectively. More specifically,  for some arbitrary operator $X\in\mathfrak{H}(\mathcal{H})$,
\begin{equation}\label{Xfull}
X=\begin{pmatrix}
 x_{11} & x_{12} & x_{13}\\
 x_{21} & x_{22} & x_{23} \\
  x_{31} & x_{32} & x_{33}
\end{pmatrix} ,
\end{equation}
we have
\begin{equation}\label{Xt1}
\Lambda_{t_1}(X)=\mathcal{E}_1(X)=\begin{pmatrix}
x_{11} &      0        &0 \\
    0        & x_{22} & 0\\
    0        &       0       & x_{33}
\end{pmatrix}.
\end{equation}
For  $t_1\leq t < t_2$, $\Lambda_{t}(X)=\Gamma^{(2)}_{(t-t_1)/(t_2-t_1)}\mathcal{E}_1(X)$,  with $\Gamma^{(2)}_{(t-t_1)/(t_2-t_1)}$ a convex combination of the identity and $\mathcal{E}_2$. Since $\mathcal{E}_2$ is trace-preserving on $\Im(\mathcal{E}_1)$, as
\begin{equation}\label{Xt2}
\mathcal{E}_2\mathcal{E}_1(X)=\begin{pmatrix}
 x_{11} & 0 & 0\\
 0 & x_{22}+x_{33} & 0\\
 0  & 0   & 0
\end{pmatrix} ,
\end{equation}
it turns out that $\Lambda_{t}$ is trace preserving for $t_1\leq t < t_2$ too.  In a completely similar fashion one proves trace preservation for $t_2\leq t < t_3$.  Finally,  we end up with
\begin{equation}\label{Xt3}
\Lambda_{t_3}(X)=\mathcal{E}_3\mathcal{E}_2\mathcal{E}_1(X)=\frac12\begin{pmatrix}
 x_{11} & 0 & 0\\
 0 & x_{22}+x_{33} & 0\\
 0  & 0   & x_{11}+x_{22}+x_{33}
\end{pmatrix} ,
\end{equation}
so that using \eqref{Gamma4} for $t_3\leq t \leq t_4$,  
\[ 
\Tr[\Lambda_{t}(X)]=\Tr\big[\Gamma^{(4)}_{(t-t_3)/(t_4-t_3)}\mathcal{E}_3\mathcal{E}_2\mathcal{E}_1(X)\big]= 2\langle 1| \mathcal{E}_3\mathcal{E}_2\mathcal{E}_1(X)| 1\rangle+2\langle 2| \mathcal{E}_3\mathcal{E}_2\mathcal{E}_1(X)| 2\rangle=x_{11}+x_{22}+x_{33}=\Tr(X).
\]
\end{proof}

On the other hand, $\Lambda_t$ is CP for any $t$, as it is a composition of CP maps. Thus, we arrive at the following corollary.

\begin{cor} The family $\{\Lambda_t\}_{t\in[0,t_4]}$ in \eqref{c-exem} forms a dynamical map.
\end{cor}

\begin{pro} The dynamical map in \eqref{c-exem} is not P-divisible.
\end{pro}

\begin{proof} Indeed, note that 
\begin{equation}
\begin{cases}
\Lambda_{t_3}(|1\rangle\langle 1|)=\mathcal{E}_3\mathcal{E}_2\mathcal{E}_1(|1\rangle \langle 1|)=\rho_A,\\
\Lambda_{t_3}(|2\rangle\langle 2|)=\mathcal{E}_3\mathcal{E}_2\mathcal{E}_1(|2\rangle \langle 2|)=\rho_B,
\end{cases}
\end{equation}
and
\begin{equation}
\begin{cases}
\Lambda_{t_4}(|1\rangle\langle 1|)=\Gamma^{(4)}_1\mathcal{E}_3\mathcal{E}_2\mathcal{E}_1(|1\rangle \langle 1|)=\Gamma^{(4)}_1(\rho_A)=|1\rangle \langle 1|,\\
\Lambda_{t_4}(|2\rangle\langle 2|)=\Gamma^{(4)}_1\mathcal{E}_3\mathcal{E}_2\mathcal{E}_1(|2\rangle \langle 2|)=\Gamma^{(4)}_1(\rho_B)=|\theta\rangle \langle \theta|.
\end{cases}
\end{equation}
Therefore, the positive and trace preserving map $\Lambda_{t,s}$ in \eqref{Divisibledecomposition} for $t=t_4$ and $s=t_3$ must satisfy 
\begin{equation}
\begin{cases}
\Lambda_{t_4,t_3}(\rho_A)=|1\rangle \langle 1|,\\
\Lambda_{t_4,t_3}(\rho_B)=|\theta\rangle \langle \theta|.
\end{cases}
\end{equation}
Since both final states are pure but $\rho_A$ and $\rho_B$ are not,  the positivity and trace preserving condition imposes that $\Lambda_{t_4,t_3}$ has to map all rank one projections in the support of $\rho_A$ into $|1\rangle \langle 1|$, and all in the support of $\rho_B$ into $|\theta\rangle \langle \theta|$.  However,  $|3\rangle \langle 3|$ is a rank one projector in the support of both $\rho_A$ and $\rho_B$, hence such a positive a trace preserving map $\Lambda_{t_4,t_3}$ does not exist unless $|1\rangle \langle 1|=|\theta\rangle \langle \theta|$, which is not true for $\theta\neq (2n+1)\pi/2$ with $n\in\mathds{N}$.

\end{proof}

Despite the dynamical map in \eqref{c-exem} is not P-divisible (so neither CP-divisible), it presents an unidirectional flow of information for some values of $\theta$.

\begin{pro} \label{MainProposition} If $\theta\in[\sqrt{2},\pi/2]$,  the dynamical map $\{\Lambda_t\}_{t\in[0,t_4]}$ in \eqref{c-exem} is monotonically contractive 
\begin{equation}\label{monNorm}
\frac{d}{dt}\|\Lambda_t X\|_1\leq 0, \quad X\in\mathfrak{H}(\mathcal{H}).
\end{equation}
\end{pro}

\begin{proof} We shall analyze every interval separately: 
\begin{itemize}
\item For $0\leq t< t_1$,  $\Lambda_t$ is given by an invertible CP-divisible map, as $\gamma(s)$  is taken to be a positive bounded function, so \eqref{monNorm} follows \cite{CKR}.
\item For $t_1\leq t< t_2$,  using \eqref{Xt1} and \eqref{Gamma2}, we find
\begin{equation}
\|\Lambda_{t} X\|_1=|x_{11}|+\big|x_{22}+\big[1-\ee^{-f_1(t)}\big]x_{33}\big|+\ee^{-f_1(t)}|x_{33}|.
\end{equation}
Here,  $f_1[(t-t_1)/(t_2-t_1)]$ has been rewritten as $f_1(t)$,  to avoid unnecessarily nasty expressions.  Therefore, for some small $\epsilon$,
\begin{align}
\|\Lambda_{t+\epsilon} X\|_1-\|\Lambda_{t} X\|_1&=\big|x_{22}+\big[1-\ee^{-f_1(t+\epsilon)}\big]x_{33}\big|-\big|x_{22}+\big[1-\ee^{-f_1(t)}\big]x_{33}\big|+\big[\ee^{-f_1(t+\epsilon)}-\ee^{-f_1(t)}\big]|x_{33}| \nonumber\\
&=\big|x_{22}+\big\{1-\ee^{-f_1(t)}[1-f'(t)\epsilon+\mathcal{O}(\epsilon^2)] \big\}x_{33}\big| \nonumber\\
-\big|x_{22}+\big[1-\ee^{-f_1(t)}\big]x_{33}\big| +[-f_1'(t)\ee^{-f_1(t)}\epsilon+\mathcal{O}(\epsilon^2)]|x_{33}| \span\omit \nonumber\\
&\leq\big|\ee^{-f_1(t)}f'(t)\epsilon x_{33}\big|-f_1'(t)\ee^{-f_1(t)}\epsilon|x_{33}|+\mathcal{O}(\epsilon^2)= \mathcal{O}(\epsilon^2),
\end{align}
where the bound follows from the triangle inequality and we have used that  $f_1(t)$ is monotonically increasing (so $f'_1(t)\geq0$) in the last step.  Thus,  
\begin{equation}
\lim_{\epsilon\downarrow 0} \frac{\|\Lambda_{t+\epsilon} X\|_1-\|\Lambda_{t} X\|_1}{\epsilon}\leq 0,
\end{equation}
and \eqref{monNorm} follows. 
\item For $t_2\leq t< t_3$,  using \eqref{Xt2} and \eqref{Gamma3}, we obtain
\begin{equation}
\|\Lambda_{t} X\|_1=\frac{1+\ee^{-f_2(t)}}{2}(|x_{11}|+|x_{22}+x_{33}|)+\frac{1-\ee^{-f_2(t)}}{2}|x_{11}+x_{22}+x_{33}|,
\end{equation}
where we have used again the shortcut notation $f_2(t)$ for $f_2[(t-t_2)/(t_3-t_2)]$ .  The condition \eqref{monNorm} is now satisfied because
\begin{align}
\|\Lambda_{t+\epsilon} X\|_1-\|\Lambda_{t} X\|_1&=\frac{\ee^{-f_2(t+\epsilon)}-\ee^{-f_2(t)}}{2}(|x_{11}|+|x_{22}+x_{33}|-|x_{11}+x_{22}+x_{33}|)\nonumber\\
&=\frac{\epsilon f_2'(t)\ee^{-f_2(t)}}{2}(|x_{11}+x_{22}+x_{33}|-|x_{11}|-|x_{22}+x_{33}|)+\mathcal{O}(\epsilon^2)\leq \mathcal{O}(\epsilon^2),
\end{align}
as a consequence of $f_2'(t)\geq 0$ and the triangle inequality, similarly as before.
\item For $t_3\leq t\leq t_4$,  the situation is more complicated.  First of all, since $\Im[\mathcal{E}_3\mathcal{E}_2\mathcal{E}_1]$ is spanned by $\rho_A$ and $\rho_B$,  we have to prove contractivity just for operators with the form $X= (\rho_A-\lambda\rho_B)$.  Indeed, linearity of $\Gamma_\tau^{(4)}$ leads to
\begin{equation}
\|\Lambda_{t} X\|_1=\|\Gamma^{(4)}_{(t-t_3)/(t_4-t_3)}\mathcal{E}_3\mathcal{E}_2\mathcal{E}_1(X)\|_1=\|\Gamma^{(4)}_{(t-t_3)/(t_4-t_3)} (\lambda_A \rho_A+\lambda_B \rho_B)\|_1= |\lambda_A|  \|\Gamma^{(4)}_{(t-t_3)/(t_4-t_3)}(\rho_A-\lambda\rho_B)\|_1
\end{equation}
with $\lambda:=-\lambda_B/\lambda_A$.  Moreover,  since $\tau=(t-t_3)/(t_4-t_3)$  increases monotonically with $t$, we can focus on proving  
\begin{equation}
\frac{d}{d\tau }\|\Gamma_\tau^{(4)}(\rho_A-\lambda\rho_B)\|_1\leq 0.
\end{equation}

Now, for $\lambda\leq 0$ the trace norm remains trivially constant  as $\Gamma_\tau^{(4)}$ is completely positive and trace preserving on $\Im[\mathcal{E}_3\mathcal{E}_2\mathcal{E}_1]$. Namely, $\rho_A-\lambda\rho_B$ is a positive operator which is mapped to another positive operator with the same trace.  For $\lambda\geq0$,  a straightforward calculation yields
\begin{equation}
\|\Gamma_\tau^{(4)}(\rho_A-\lambda\rho_B)\|_1=\frac{1}{2} \left[\left(1-\tau^2\right)\left| (\lambda -1)\right| +\left(1+\tau^2\right) \sqrt{1+\lambda ^2+2 \lambda  \cos (2 \theta  \tau)}\right].
\end{equation}
Differentiation with respect to $\tau$ gives
\begin{equation}\label{dGamma4}
\frac{d}{d\tau }\|\Gamma_\tau^{(4)}(\rho_A-\lambda\rho_B)\|_1=\tau \left[-|\lambda-1| +\sqrt{1+\lambda ^2+2 \lambda  \cos (2\theta \tau)}\right]-\frac{\lambda \theta \left(1+\tau^2\right) \sin (2\theta \tau)}{\sqrt{1+\lambda ^2+2 \lambda  \cos (2\theta \tau)}}.
\end{equation}

Consider first the case $\lambda\geq1$ where
\begin{equation}\label{dGamma5}
\frac{d}{d\tau }\|\Gamma_\tau^{(4)}(\rho_A-\lambda\rho_B)\|_1=\tau \left[1-\lambda+\sqrt{1+\lambda ^2+2 \lambda  \cos (2\theta \tau)}\right]-\frac{\lambda \theta \left(1+\tau^2\right) \sin (2\theta \tau)}{\sqrt{1+\lambda ^2+2 \lambda  \cos (2\theta \tau)}}.
\end{equation}
For a fixed value of $\theta\tau$, the $\lambda$-derivative of the term between square brackets of \eqref{dGamma5} is
\[
 -1+\frac{\lambda+\cos (2\theta \tau)}{\sqrt{1+\lambda ^2+2 \lambda  \cos (2\theta \tau)}}=-1+\sqrt{\frac{[\lambda+\cos (2\theta \tau)]^2}{1+\lambda ^2+2 \lambda  \cos (2\theta \tau)}}=-1+\sqrt{\frac{\cos^2 (2\theta \tau)+\lambda^2+2\lambda\cos (2\theta \tau)}{1+\lambda ^2+2 \lambda  \cos (2\theta \tau)}}\leq0
\]
and the function is monotonically decreasing with $\lambda$.  As a result, the first term of \eqref{dGamma5} has a maximum for $\lambda=1$. On the other hand, using also that
\begin{equation}
\frac{\lambda}{{\sqrt{1+\lambda ^2+2 \lambda  \cos (2\theta \tau)}}}\geq\frac{\lambda}{1+\lambda}\geq \frac12,
\end{equation}
the equation \eqref{dGamma5} can be upper bounded by
\begin{equation}\label{bound_dGamma}
\frac{d}{d\tau }\|\Gamma_\tau^{(4)}(\rho_A-\lambda\rho_B)\|_1\leq \tau\sqrt{2+2\cos(2\theta\tau)}-(1+\tau^2)\frac{\theta}{2}\sin(2\theta\tau).
\end{equation}
Now, for any angle $\alpha\in[0,\pi]$, $\sqrt{1+\cos\alpha}=\sqrt{2}\cos(\frac{\alpha}{2})$, and then the double-angle formula for the sine allows us to rewrite equation \eqref{bound_dGamma} for $\theta\in[0,\pi/2]$ as
\begin{equation}\label{bound_dGamma2}
\frac{d}{d\tau }\|\Gamma_\tau^{(4)}(\rho_A-\lambda\rho_B)\|_1\leq \cos(\theta\tau) [2\tau -(1+\tau^2)\theta\sin(\theta\tau)].
\end{equation}
Since $\cos(\theta\tau)\geq0$ for $\theta\in[0,\pi/2]$, we shall focus on the term between square brackets of \eqref{bound_dGamma2}. The sine of any angle is larger than its 3rd order Taylor expansion, $\sin \alpha\geq \alpha -\tfrac16 \alpha^3 $, so
\begin{equation}\label{polynomial}
[2\tau -(1+\tau^2)\theta\sin(\theta\tau)]\leq 2\tau -(1+\tau^2)\theta(\theta\tau-\tfrac16\theta^3\tau^3)=(2-\theta^2)\tau-(\theta^2-\tfrac16\theta^4)\tau^3+\tfrac16\theta^4\tau^5\leq (2-\theta^2)\tau-(\theta^2-\tfrac13\theta^4)\tau^3,
\end{equation}
where we have used that $\tau^3\geq \tau^5$ for $\tau\in[0,1]$ in the last step.  It is easily found that this polynomial only has a real root at $\tau=0$ for $\theta\geq\sqrt{2}$.  Since for small $\tau$ this condition clearly makes negative this polynomial,  the same sign remains for any $\tau>0$ as long as $\theta\in[\sqrt{2},\frac{\pi}2]$.

In the case that  $0<\lambda<1$,  Eq.  \eqref{dGamma4} becomes
\begin{align}\label{dGamma6}
\frac{d}{d\tau }\|\Gamma_\tau^{(4)}(\rho_A-\lambda\rho_B)\|_1=\tau \left[\lambda-1 +\sqrt{1+\lambda ^2+2 \lambda  \cos (2\theta \tau)}\right]-\frac{\lambda \theta \left(1+\tau^2\right) \sin (2\theta \tau)}{\sqrt{1+\lambda ^2+2 \lambda  \cos (2\theta \tau)}}\nonumber \\
=\lambda\left\{\tau \left[1-\frac{1}{\lambda} +\sqrt{\frac{1}{\lambda^2}+1+\frac{2}{ \lambda}  \cos (2\theta \tau)}\right]-\frac{\theta \left(1+\tau^2\right) \sin (2\theta \tau)}{\lambda \sqrt{\frac{1}{\lambda^2}+1+\frac{2}{ \lambda}  \cos (2\theta \tau)}}\right\}.
\end{align}
The term between curly brackets is the same as the right hand side of \eqref{dGamma5} under the change $\lambda\to1/\lambda$.  Since we have proven that the latter is negative for $\lambda\geq1$ if $\theta\in[\sqrt{2},\frac{\pi}2]$,  the former is also negative for $0<\lambda<1$.  This proves the negativity of \eqref{dGamma6} for $\theta\in[\sqrt{2},\frac{\pi}2]$.  
\end{itemize}

The interval $\theta\in[\sqrt{2},\frac{\pi}2]$ is in fact tight, as for $\theta>\pi/2$ Eq. \eqref{dGamma4} becomes positive for $\tau=1$,  and for $\theta<\sqrt{2}$ it becomes positive for $\tau$ close to 0.
\end{proof}

It is worth mentioning that this counterexample can be slightly modified to ensure that the dynamical map $\{\Lambda_t\}_{t\in[0,t_4]}$ has continuous time-derivative. To this end, it is enough to take e.g. $\gamma(s)=(1-s)^{-1}$ and $f_{1}(\tau)=f_{2}(\tau)=\tau^2/(1-\tau)$ to obtain zero time-derivative of $\{\Lambda_t\}_{t\in[0,t_4]}$ in $t\in\{t_1,t_2,t_3\}$ (in the rest of the points continuity of the derivative trivially follows). In order to ensure a zero time-derivative in $t_3$ also from the right direction, we can simply take $\Gamma_{\tau}^{(4)}$ as in \eqref{Gamma4} but with the change $\tau\to \tau^\delta$ with $\delta>1$ in the argument of the exponential term. Following the same steps as in the proof of Proposition \ref{MainProposition}, we see that this change induces a $\delta-$deformation of Eq. \eqref{polynomial}. Thus, this remains negative for $\delta$ sufficiently close to (but larger than) 1 provided that $\theta$ is sufficiently far from the left end of $[\sqrt{2},\pi/2]$.

\section{Conclusions}
We have constructed a non-P-divisible dynamical map which presents unidirectional information flow as measured by trace norm contractivity. This counterexample closes an open problem regarding the extension of equivalence theorems between two of the most studied Markovianity conditions to noninvertible dynamical maps. Particularly, it prevents from a extension of Theorem \ref{TheoremCRS} for completely arbitrary dynamics. The result is also useful to illustrate the intricate structure of noninvertible dynamical maps. Nevertheless, there are at least two immediate questions which remains unsolved. On the one hand, one can wonder whether or not the image nonincreasing dynamical maps (or their unitarily equivalent) considered in Theorem \ref{TheoremCRS} and Corollary \ref{CorollaryCRS} form the larger class of dynamical maps which allow for the equivalence between both Markovinity notions. On the other hand, it is not known whether a counterexample to the result of Theorem \ref{TheoremCC} for a two-dimensional Hilbert space can be obtained in larger dimensional spaces.

\section*{Acknowledgements}
This work has been supported by the Spanish MINECO grants MINECO/FEDER
Project PGC2018-099169-B-I00 FIS-2018 and from CAM/FEDER Project No. S2018/TCS-4342
(QUITEMAD-CM). The author is grateful to Prof. Dariusz Chru\'sci\'nski for illuminating discussions on this topic, and to Prof. Andrzej Kossakowski, to whose memory this work is dedicated. His humility, extreme kindness and perpetual good humor were even greater achievements than his extraordinary scientific legacy.

\end{document}